\documentclass[journal]{IEEEtran}
\ifCLASSINFOpdf
\else
\fi
\usepackage{amsmath,amsthm}
\usepackage{txfonts}
\newtheorem {theorem}{Theorem}
\usepackage{graphicx}
\usepackage{subfigure}
\usepackage{cleveref}
\usepackage{comment}
\Crefname{figure}{Fig.}{Figs.}%
\begin{document}
%
\title{MF is always superior to CEM}

\author{Xiurui~Geng,~ Luyan Ji,~Weitun~Yang,~Fuxiang~Wang,~Yongchao~Zhao
\thanks{X. Geng, W. Yang and Y. Zhao are with the Key Laboratory of Technology in Geo-Spatial information Processing and Application System, Institute of Electronics, Chinese Academy of Sciences, Beijing 100190, China e-mail: gengxr@sina.com.}
\thanks{L. Ji is the Ministry of Education Key Laboratory for Earth System Modeling, Center for Earth System Science, Tsinghua University.}
\thanks{F. Wang is with the School of Electronics and Information Engineering , Beihang University.}
\thanks{Manuscript received ?, ?; revised ?, ?.This work was supported in part by the National Program on Key Basic Research Project (973 Program) under Grant 2015CB953701.}}

\markboth{Journal of ?,~Vol.~?, No.~?, ?~?}%
{Geng \MakeLowercase{\textit{et al.}}: MF is always superior to CEM}
%

\maketitle

\begin{abstract}
The constrained energy minimization (CEM) and matched filter (MF) are two most frequently used target detection algorithms in the remotely sensed community. In this paper, we first introduce an augmented CEM (ACEM) by adding an all-one band. According to a recently published conclusion that CEM can always achieve a better performance by adding any linearly independent bands, ACEM is better than CEM. Further, we prove that ACEM is mathematically equivalent to MF.  As a result, we can conclude that the classical matched filter (MF) is always superior to the CEM operator.
\end{abstract}


%
\IEEEpeerreviewmaketitle

\section{Introduction}
%
%
%
%
\IEEEPARstart{I}{n} the field of target detection, the constrained energy minimization (CEM) \cite{Joseph1993Detection} has been widely used in various applications, such as geological survey\cite{FARRAND199764, Alam2008Mine}, agriculture management \cite{Monaco2014High} and medical image processing\cite{5604768, Lin2010721}, and further developed for real-time processing \cite{Chang2003FPGA,917889,progreesiveCEM}. Recently, it has received more and more attentions. One direct proof is that the most widely used remote sensing software, environment for visualizing images (ENVI) has included CEM since Version 4.6.

CEM is originally derived from a linearly constrained minimum variance adaptive beam-forming in the field of signal processing. It keeps the output energy of the target as a constant while suppressing the output of the background to a minimum level. Recently, researchers have started to explore the influence of data dimensionality on hyperspectral target detection algorithms, and noticed that more bands can help to improve the detection result\cite{Chang1999Generalized,5653628,Chen2011Effects}. Geng et.al\cite{6786322} further proved that adding any band linearly independent of original image will always lead to the performance increase of CEM. Therefore, by adding  linearly independent and target-benefit bands, CEM can be applied to multispectral target detection\cite{w7020794}. 

Besides the energy criterion used by CEM, another widely used criterion in the field of target detection is  maximum likelihood criterion, which is represented by the match filter (MF)  detector\cite{Manolakis2003Hyperspectral,974724,4217133}. MF is an commonly used technique in the field of communication and signal processing applications\cite{7071091,1511829}. It has been further developed for hyperspectral target detection, and thus widely applied in the field of remote sensing \cite{Manolakis2000Comparative, ParkerWilliams2002446,5985518}. MF has been embedded in ENVI since a very early version. It is the optimum-Neyman-Pearson detector when the target and background classes are normally distributed and their covariance matrices are assumed to be equal\cite{1295218}.

Though the CEM and MF detectors are based on very different theories, their mathematical expressions are similar, except that the MF detector needs the data to be centralized first. Usually, the difference of MF and CEM in mathematical form is easily noticed\cite{Manolakis2003Hyperspectral,1295218}, but their performance is seldom compared. Therefore, which of them can achieve a better performance is still an unsolved problem in theory.

In this paper, according to conclusion reported in the reference\cite{6786322}, we prove that MF is always superior to CEM. That is to say, of the two benchmark target detection methods, CEM can now be considered obsolete. 

\section{Background}
In this section, we will first introduce the expression of MF and CEM detectors, and then briefly describe the influence of band numbers on CEM.

\subsection{MF}

According to the Neyman-Pearson criterion, the optimum decision strategy can be achieved by maximizing the probability of detection while keeping the probability of false alarm under a certain value\cite{Manolakis2003Hyperspectral}.

Assume that the observed data matrix is given by $ \mathbf{X}= \left[\mathbf{r}_1,\mathbf{r}_2,\dots,\mathbf{r}_N \right] $ , where  $\mathbf{r}_i=\left[r_{i1},r_{i2},\dots,r_{iL}\right]^T $  for $ 1\le i \le N $ is a sample pixel vector, $ N $ is the total number of pixels, and $ L $ is the number of bands. Suppose that the desired signature $ \mathbf{d} $ is also known. Then, the normalized expression of an MF detector can be written as \cite{Manolakis2003Hyperspectral,1295218}
\begin{equation}
	\mathbf{w}_{MF}=c_{MF}\mathbf{K}^{-1}\left(\mathbf{d}-\mathbf{m}\right)=\frac{\mathbf{K}^{-1}\left(\mathbf{d}-\mathbf{m}\right)}{\left(\mathbf{d}-\mathbf{m}\right)^T\mathbf{K}^{-1}\left(\mathbf{d}-\mathbf{m}\right)}
	\label{w_MF}
\end{equation}
where $ \mathbf{m}=\left(\sum_{i=1}^{N}\mathbf{r}_i\right)/N $  is the mean vector, $ \mathbf{K}=\left[\sum_{i=1}^{N}\left(\mathbf{r}_i-\mathbf{m}\right)\left(\mathbf{r}_i-\mathbf{m}\right)^T\right]/N $  is the covariance matrix, $ c_{MF}= 1/\left[\left(\mathbf{d}-\mathbf{m}\right)^T\mathbf{K}^{-1}\left(\mathbf{d}-\mathbf{m}\right)\right] $ is a scalar.

\subsection{CEM}

CEM is proposed by Harsanyi in 1993, which is originally derived from the linearly constrained minimized variance adoptive beam-forming in the field of digital signal processing. It uses a finite impulse response (FIR) filter to constrain the desired signature by a specific gain while minimizing the filter output energy \cite{Joseph1993Detection}. The objective of CEM is to design an FIR linear filter $\mathbf{w}=\left[w_{1},w_{2},\dots,w_{L}\right]^T $ to minimize the filter output power subject to the constraint, $ \mathbf{d}^T\mathbf{w}=\sum_{l=1}^{L}d_lw_l=1 $. Then the problem yields 

\begin{equation}
	\left\{
	\begin{aligned}
		&\min \limits_{\mathbf{w}}\frac{1}{N}\left(\sum_{i=1}^{N}y_i^2\right)=\min \limits_{\mathbf{w}}\mathbf{w}^T\mathbf{R}\mathbf{w} \\
		&\mathbf{d}^T\mathbf{w}=1 
	\end{aligned}\right.,
	\label{CEM_obj}
\end{equation}
where $ y_i=\mathbf{w}^T\mathbf{r}_i $ and $ \mathbf{R}=\left(\sum_{i=1}^{N}\mathbf{r}_i\mathbf{r}_i^T\right)/N $, which is firstly referred to as sample correlation matrix by Harsanyi \cite{Joseph1993Detection,FARRAND199764}, and later renamed autocorrelation matrix by some other researchers \cite{Liu1999Generalized,Chang1999Generalized, 1295199, 843007}. In this paper, we will adopt Harsanyi's nomination. The solution to this constrained minimization problem (\ref{CEM_obj}) is the CEM operator, $ \mathbf{w}_{CEM} $   given by \cite{Joseph1993Detection}
\begin{equation}
	\mathbf{w}_{CEM}=\frac{\mathbf{R}^{-1}\mathbf{d}}{\mathbf{d}^T\mathbf{R}^{-1}\mathbf{d}}.
	\label{w_CEM}
\end{equation}

The CEM detector has a very similar form to the MF detector. The only difference is whether we remove mean vector from all the data pixels (including the target signature) in advance. 

\subsection{CEM: more bands, better performance}

Geng et.al\cite{6786322} prove that adding any band linearly independent of the original ones will improve the detection performance of CEM. In other words, the performance of CEM will decrease when removing any bands of the data. 
Suppose $ \Omega \subset \{ 1,2,\dots,L \} $  is an arbitrary subset of the band index set $ \{ 1,2,\dots,L \}$ ; $ \mathbf{R}_\Omega $ and $ \mathbf{d}_\Omega $ are the corresponding sample correlation matrix and target spectral vector respectively. The theorem of CEM on the number of bands is given as follows\cite{6786322}:
\begin{theorem}
	the output energy from full bands is always less than that from the partial bands, i.e.
	\begin{equation}
		\frac{1}{\mathbf{d}^T\mathbf{R}^{-1}\mathbf{d}} < \frac{1}{\mathbf{d}^T_\Omega\mathbf{R}^{-1}_\Omega\mathbf{d}_\Omega}.
	\end{equation} 
	\label{Theorem1}                            
\end{theorem} 

Based on this theorem, we will present an augmented CEM algorithm and then use it as a bridge to prove that MF is always superior to CEM in the next section.

\section{MF is always superior to CEM}
In this section, we will firstly produce an auxiliary CEM detector by adding an all-one band to the original ones and then demonstrate its equivalence with MF detector.

\subsection{the Augmented CEM}
Theorem \ref{Theorem1} indicates that adding any band that is linearly independent of the original ones will improve the result of CEM from the angle of output energy. Thus, in this section, we will add an all-one band to the data set, and the corresponding algorithm is named the augmented CEM (ACEM). Based on Theorem \ref{Theorem1}, ACEM can achieve a better performance than CEM.

First of all, add an all-one band to $ \mathbf{X} $, and we can get the augmented data matrix, $ \tilde{\mathbf{X}}=\left[\begin{array}{cc} \mathbf{X} \\ \mathbf{1}^T  \end{array}  \right]  $ (where $ \mathbf{1}=\left[1,1,\dots,1\right]^T $ is an $N$-dimensional column vector). Accordingly, the augmented target vector becomes $ \tilde{\mathbf{d}}=\left[\begin{array}{cc} \mathbf{d} \\ 1  \end{array}  \right] $. Similar to CEM, the ACEM detector, $ \mathbf{w}_{ACEM} $, which is an $ \left(L+1\right) $-dimensional column vector can be calculated as 
\begin{equation}
	\mathbf{w}_{ACEM}= c_{ACEM}\tilde{\mathbf{R}}^{-1}\tilde{\mathbf{d}}=\frac{\tilde{\mathbf{R}}^{-1}\tilde{\mathbf{d}}}{\tilde{\mathbf{d}}^T\tilde{\mathbf{R}}^{-1}\tilde{\mathbf{d}}}.
	\label{w_ACEM}
\end{equation}
where $ \tilde{\mathbf{R}}=\left(\tilde{\mathbf{X}}\tilde{\mathbf{X}}^T\right)/N $ is an $ \left(L+1\right) \times \left(L +1 \right) $ matrix, and $ c_{ACEM}=1/\left(\tilde{\mathbf{d}}^T\tilde{\mathbf{R}}^{-1}\tilde{\mathbf{d}}\right) $ is a scalar. 

Next, we aim to prove the equivalence between the ACEM and MF detector. Since the last band of $ \tilde{\mathbf{X}} $ is an constant band, the detection result of ACEM is only determined by the first $ L $ elements of $ \mathbf{w}_{ACEM} $. Therefore, we only need to prove the equivalence between $ \mathbf{w}_{ACEM\left(1:L\right)} $ and $ \mathbf{w}_{MF} $, where $ \mathbf{w}_{ACEM\left(1:L\right)} $ denotes the first $ L $ elements of $ \mathbf{w}_{ACEM} $.

\subsection{ACEM is equivalent to MF}
In this section, we will demonstrate the equivalence between ACEM and MF in the following theorem:

\begin{theorem}
	The ACEM detector is equivalent to the MF detector. That is, there exists a constant $ c $ such that 
	\begin{equation}
		\mathbf{w}_{ACEM\left(1:L\right)}=c \mathbf{w}_{MF}.
	\end{equation}
	\label{Theorem2}
\end{theorem}
\begin{proof}
	The covariance matrix $ \mathbf{K} $ in (\ref{w_MF}) can be expressed by the sample correlation matrix $ \mathbf{R} $ and mean vector $ \mathbf{m} $ as
	
	\begin{equation}
		\mathbf{K}=\mathbf{R}-\mathbf{m}\mathbf{m}^T.
	\end{equation}
	
	Using Sherman--Morrison formula \cite{Sherman1949Adjustment}, the inverse matrix of $ \mathbf{K} $ can be calculated by 
	\begin{equation}
		\mathbf{K}^{-1}=\mathbf{R}^{-1}+b_1\mathbf{R}^{-1}\mathbf{m}\mathbf{m}^T\mathbf{R}^{-1},
		\label{K_inv}
	\end{equation}
	where the parameter $ b_1=1/\left(1-\mathbf{m}^T\mathbf{R}^{-1}\mathbf{m}\right) $. Substitute (\ref{K_inv}) into (\ref{w_MF}), and after a little algebra we can get the MF detector 
	\begin{equation}
		\mathbf{w}_{MF}=c_{MF}\left[\mathbf{R}^{-1}\mathbf{d}+b_1(b_2-1)\mathbf{R}^{-1}\mathbf{m}\right], 
		\label{w_MF3}
	\end{equation}
	where the parameter $ b_2= \mathbf{m}^T\mathbf{R}^{-1}\mathbf{d} $. Similarly, the sample correlation matrix $ \tilde{\mathbf{R}} $ in (\ref{w_ACEM}) can also be expressed by $ \mathbf{R} $ and $ \mathbf{m} $ as
	\begin{equation}
		\tilde{\mathbf{R}}=\left[\begin{array}{cc}
			\mathbf{R} & \mathbf{m} \\
			\mathbf{m}^T & 1
		\end{array}\right].
	\end{equation}
	Then, we expand the inversion of $ \tilde{\mathbf{R}} $ as 
	\begin{equation}
		\tilde{\mathbf{R}}^{-1}=\left[\begin{array}{cc}
			\mathbf{R}^{-1}+b_1\mathbf{R}^{-1}\mathbf{m}\mathbf{m}^T\mathbf{R}^{-1} & -b_1\mathbf{R}^{-1}\mathbf{m} \\
			-b_1\mathbf{m}^T\mathbf{R}^{-1} & b_1
		\end{array}\right].
		\label{AR_1}
	\end{equation}
	Substituting (\ref{AR_1}) into (\ref{w_ACEM}), we can have
	\begin{equation}
		\mathbf{w}_{ACEM}= c_{ACEM}\left[ \begin{array}{cc}
			\mathbf{R}^{-1}\mathbf{d}+b_1\left(b_2-1\right)\mathbf{R}^{-1}\mathbf{m}\\
			-b_1\left(b_2-1\right)
		\end{array}
		\right]. 
		\label{w_ACEM30}
	\end{equation}
	Thus, 
	\begin{equation}
		\mathbf{w}_{ACEM\left(1:L\right)}=c_{ACEM}\left[\mathbf{R}^{-1}\mathbf{d}+b_1\left(b_2-1\right)\mathbf{R}^{-1}\mathbf{m}\right].
		\label{w_ACEM3}
	\end{equation}
	Compared (\ref{w_ACEM3}) with (\ref{w_MF3}), we can have
	\begin{equation}
		\mathbf{w}_{ACEM\left(1:L\right)}=c \mathbf{w}_{MF}~~~\text{ with }~~~c=c_{ACEM}/c_{MF}.
	\end{equation}	 
\end{proof}

Since the all-one band can not be generally linearly expressed by the original data bands, according to Theorem \ref{Theorem1}, ACEM always obtain a better performance than CEM. Based on Theorem \ref{Theorem2}, ACEM is equivalent to MF. Therefore, it can be concluded that MF is always superior to CEM. 

As we can know, a constant band will not increase the separability between the target and background. It indicates that adding such a band should not bring any benefit to the target detection result. Yet, based on Theorem \ref{Theorem1}, the constant band can actually improve the performance of CEM. Clearly, here emerges a paradox! 
The reason, we think, is that the energy criterion used by CEM is problematic (or not perfect). In contrast, MF does not have this problem because in MF we need to move all the data points to the data center and the all-one band will then become a zero band, which is linearly correlated to all the original bands. That's why MF can avoid the influence of constant band. In all, MF can always surpass CEM, so CEM can now be considered as a redundant one.

\section{Conclusion} 

MF is the best target detector from the perspective of maximum likelihood, while CEM is the representative one from the perspective of energy. Usually, it is difficult to theoretically compare algorithms developed from different criteria, so MF and CEM are considered as two benchmark methods in target detection and both are embedded in the ENVI software, which is one of the most frequently used software packages in the remote sensing community. In this study, we first introduce an auxiliary method, called the augmented CEM (ACEM), which is implemented by adding an all-one band to the original data. According to the theorem in Ref \cite{6786322}, we can derive that ACEM can always receive a better performance than CEM in the sense of output energy criterion. Next, we prove the equivalence between ACEM and MF, which indirectly demonstrates that MF is always superior to CEM. Thus, we suggest that the classical target detection CEM should be considered redundant. Moreover, the energy criterion used by CEM is problematic since it will lead to a paradox, so in the future, we will put emphasis on finding a more reasonable criterion for target detection.

\bibliographystyle{ieeetr}
\bibliography{Reference}

\end{document}